\documentclass{article}
\usepackage[utf8]{inputenc}
\usepackage{authblk}
\usepackage{setspace}
\usepackage[margin=1.25in]{geometry}
\usepackage{graphicx}
\graphicspath{ {./figures/} }
\usepackage{subcaption}
\usepackage{amsmath}
\usepackage{amssymb}
\usepackage{amsthm}
\usepackage{bm}
\newtheorem{prop}{Proposition}
\usepackage{booktabs}
\usepackage{caption}
\usepackage{changepage}
\usepackage{float}
\usepackage{booktabs} 
\usepackage{makecell}
\usepackage{microtype}

\usepackage[style=numeric,
            citestyle=numeric-comp,
            sorting=none,
            backend=bibtex]{biblatex}
\usepackage[colorlinks=true,
            linkcolor=blue,
            citecolor=blue,
            urlcolor=black]{hyperref}

\title{A Robust Numerical Framework for Multi-Strain Epidemics}

\author[1*,2]{Fadilah Ilahi}
\author[1]{Ian Hall}
\author[1]{Thomas House}

\affil[1]{Department of Mathematics, University of Manchester, Manchester, United Kingdom.}
\affil[2]{Department of Mathematics, Universitas Islam Negeri (UIN) Sunan Gunung Djati, Bandung, Indonesia.}
\affil[$\ast$]{\href{mailto:fadilah.ilahi@manchester.ac.uk}{fadilah.ilahi@manchester.ac.uk}}

\date{}

\begin{document}
\sloppy

\maketitle

\begin{abstract}
Epidemiological models often take the form of systems of differential equations of varying levels of complexity. In multi-strain models, complexity arises from the interaction of diverse pathogen types that may exhibit different patterns of cross-immunity due to the host's immune system. Where previous work has tended to focus on theoretical analysis of this complexity, or model simplification, less effort has been given to efficient numerical exploration of full multi-strain disease dynamics. Here, we propose a computational framework for representing multi-strain disease transmission models, through a mathematically transparent hard-coded mapping between sets and integers. By mapping infection histories into Boolean matrices, the framework eliminates explicit iteration over strain combinations while preserving key biological processes, including sequential infection and partial cross-immunity. Benchmark comparisons of numerical solutions of our matrix formulation and a set-indexed implementation demonstrate substantial performance gains. Across models with increasing numbers of strains, the matrix formulation consistently reduces runtime, with a speedup exceeding two orders of magnitude for systems with 11 strains. We present an application of this approach to the case of COVID-19 in Indonesia, focusing on the replacement of the Delta variant by Omicron. We hope that our approach represents a robust computational approach to automatically generate diverse multi-strain epidemic models that can be solved in a computationally efficient manner for a variety of scientific and public-health questions. 
\end{abstract}

\section*{Introduction}

Influenza, COVID-19, and dengue are examples of viruses of significant epidemiological and public health interest where multiple strains can play an important role in the infectious disease dynamics \cite{kucharski2012original, dyson2021future}. These diseases involve various types of the pathogen, sometimes labelled serotypes, subtypes or variants as well as strains, that can pose additional challenges. Dengue, for instance, has four distinct serotypes and infection with one serotype provides only partial immunity to the others \cite{senaratne2020dengue, Woodall2014PartialEpidemiology}; secondary and tertiary infections can be more severe as a consequence of being infected by different strains from the previous infection \cite{Recker2009ImmunologicalDengue}. As well as having subtypes such as A(H1N1) and A(H3N2), influenza is constantly evolving through a processes known as antigenic drift and antigenic shift \cite{andreasen1997dynamics, boni2004influenza, ballesteros2009influenza}. And the COVID-19 pandemic saw impacts throughout the World due to the emergence and spread of variants such as Alpha, Delta and Omicron \cite{WHOvariant}.

The simplest epidemic models such as the Susceptible-Infected-Recovered (SIR) model typically focus on a single strain of a pathogen and do not seek to capture strain history within the complex dynamics of multi-strain diseases \cite{allen2008mathematical, champagne2019comparison,lazebnik2022generic}. As a consequence, they also fail to account for partial immunity that can lead to secondary, sometimes more severe, infection with the currently circulating strain. Addressing this complexity requires models that can account for the presence of multiple strains and track individual infection histories. One of the most sophisticated approaches to address this problem is the individual history-based model that has been done by Andreasen et al. (1997) \cite{andreasen1997dynamics, kucharski2016capturing}. This approach allows for detailed transmission process of how diseases spread within a population by tracking the immune status of each individual based on their history of exposure to different strains. However, this model can be computationally challenging when working with a large number of strains in the set-equation framework in which it is formulated. A common response to this challenge has been to derive simplified models that are mathematically distinct from the model of Andreasen et al., while preserving the most epidemiologically important properties for the research questions in hand \cite{gog2002dynamics, kucharski2016capturing}.

Here we develop a more computationally straightforward and efficient framework that transforms the transmission of the multi-strain disease from a traditional system of sets to a matrix-based numerical system. This method condenses the complex interactions between multiple strains and the host population into a structured model that captures key elements such as infection histories, partial cross-immunity, and strain-specific dynamics. It is important to note that the underlying epidemiological model remains \emph{mathematically} unchanged from the formulation of Andreasen et al. \cite{andreasen1997dynamics}: The primary contribution of this work is a \emph{computational} reformulation that enables the same model to be simulated substantially faster as the number of circulating strains $n$ increases without requiring modification of the biological assumptions.

In this paper, we therefore present a numerical framework for multi-strain models, designed to capture the essential dynamics of diseases like SARS-CoV-2 which involve the co-circulation of multiple variants \cite{wang2022multiscale, souto2019multiple, boni2004influenza} during their transmission, by mapping set-indexed compartments into structured vectors and indicator matrices, we reformulate a broad class of multi-strain transmission models into compact linear-algebraic expressions. This representation eliminates explicit looping over strain histories, enabling efficient numerical simulation while preserving key biological mechanisms such as partial cross-immunity and sequential infection. We note, however, that this matrix formulation is likely to be less suitable for analytical work and so provide mappings between both approaches, which are best viewed as complementary.

We then apply our framework to strain replacement dynamics \cite{martcheva2008vaccine, lythgoe2023lineage, mahroum2024vaccine}, a phenomenon in which one pathogen variant progressively displaces another due to differences in transmissibility or immune escape. Using SARS-CoV-2 Delta and Omicron variants in Indonesia as a motivating example, we show that the matrix formulation reproduces observed replacement patterns and enables direct comparison between empirical data and mechanistic model output. We further illustrate scalability by extending the model to higher-dimensional systems with multiple competing strains. By combining mechanistic epidemiological modeling with numerical implementation, this framework provides a computationally tractable bridge between theory and data for studying multi-strain epidemics.

\section*{Materials and Methods}
In this section, we develop the methods and establish the assumptions for constructing a new framework and to tackle the problems discussed in the introduction. First, we adapt the model from Andreasen et al.\ \cite{andreasen1997dynamics} in order to account for the history-based aspects of the multi-strain model. Second, we evaluate the model to compare with real data to represent the strain replacement using the newly constructed model. 

\subsection*{Model assumptions}

The multi-strain model in this research is the influential multi-strain model of Andreasen et al.\ \cite{andreasen1997dynamics}, where the infection is to be considered as history-based with partial cross-immunity. The model assumes that the population is initially homogeneous, meaning that all individuals are equally likely to be exposed to any strain of the pathogen and no individual has inherent immunity to any strain at an early stage. The transmission of each strain follows a process governed by a set of parameters such as transmission rate, recovery rate, and cross-immunity rate. The model assumes that multiple strains of the pathogen can circulate simultaneously within the population. Infection with one strain can provide some cross-immunity to other strains but this need not be sufficient to prevent infection with a new strain, i.e.~strains interact through partial cross-immunity \cite{Woodall2014PartialEpidemiology}. In addition, the population is assumed to be closed without migrations in or out of the population during the epidemic period, although births and deaths can be included \cite{lazebnik2022generic}. 

We assume the total population size is $N$ and the set of serotype indices is 
$\mathcal{K} = \{1,2,3,...,n\}$ \cite{andreasen2023seasonality}. In terms of history of infection with different strains, this can be indexed with a subset $\mathcal{J}\subseteq\mathcal{K}$. For example, for Dengue fever there are four serotypes and therefore $\mathcal{K}=\{1, 2, 3, 4\}$. This model will make use of the set of all subsets of $\mathcal{K}$, its \emph{power set}, $P(\mathcal{K})$, sometimes written $2^{\mathcal{K}}$ and of size $|P(\mathcal{K})|=2^n$.

Usually this system will be analysed as starting from a perturbation around the disease-free state, where the population with no history of infection is equal to the total population, $S_{\varnothing} = N$. Individuals then potentially experience infection with one or more (up to the maximum of $n$) strains \cite{andreasen1997dynamics, wikramaratna2010effects}.

In the transmission process, for an individual with infection history $\mathcal{J}\subseteq\mathcal{K}$ it is assumed that further infection only occurs with strains $i$ such that $i \neq \mathcal{J}$. 
The rate constant for infection with each strain denoted by $\beta_i$. The per-individual rate of transmission of the disease to those susceptible to strain $i$ is given by the force of infection $\Lambda^i$, which is determined by summing over all the hosts infections with strain $i$:
\begin{equation}
    \begin{aligned}
        \Lambda^i={\beta_i}\sum_{\mathcal{J}\subseteq{\mathcal{K}\setminus\{i\}}}^{}I_\mathcal{J}^i \, .
        \label{foi}
    \end{aligned}
\end{equation}
For second and subsequent infectious, we need to consider cross-immunity. The cross-reactivity factor for those who are infected bystrain history $\mathcal{J}$ is $\sigma_\mathcal{J}^i$. As a consequence, in the population $S_\mathcal{J}$ will be depleted at rate
$$
\sum_{i\notin\mathcal{J}}^{} \sigma_\mathcal{J}^i \Lambda^i S_\mathcal{J} \, .
$$
The number of susceptible classes will be $2^n$ while the number of infected classes will be $n2^{n-1}$ giving a total of $(n+2)2^{n-1}$ \cite{andreasen1997dynamics, kucharski2016capturing}. Based on the original model of Andreasen et al. \cite{andreasen1997dynamics, kucharski2016capturing}, it is assumed that the model includes the demography where there are births and death rates represented by $b$ and $\mu$ respectively. So the equilibrium population size will be $N=b/{\mu}$. The mathematical model for this system is then\cite{kucharski2016capturing}:
\begin{equation}
    \begin{aligned}
        \dot S_{\varnothing} &= b - \sum_{j\in\mathcal{K}}^{} \Lambda^jS_{\varnothing} - \mu S_{\varnothing},\\
        \dot I_\mathcal{J}^i &= \sigma_\mathcal{J}^i \Lambda^i S_\mathcal{J} - (\mu+\gamma)I_\mathcal{J}^i,\\
        \dot S_\mathcal{J \neq \varnothing} &= \sum_{j \in \mathcal{J}}^{} \gamma I_{\mathcal{J}\setminus \{j\}}^j - \sum_{i\notin\mathcal{J}}^{} \sigma_\mathcal{J}^i \Lambda^i S_\mathcal{J} - \mu S_\mathcal{J},
        \label{viggooriginal}
    \end{aligned}
\end{equation}
\noindent
where $\gamma$ is the rate at which infected individuals recover from infection. Note that $\mathcal{J}\neq0$ and $\sigma_\mathcal{K}^i = 0$ for all $i$. The impact cross-reactivity is first observed in the second infection and in general depends on both the history strains $\mathcal{J}$ and current strain $i$. For complex system like dengue, the combination of strains that cause antibody-dependent enhancement (ADE) is uncertain, which means that the cross-reactivity will only depends on the distinct strains in the second, third \ldots{} and $n$-th infections \cite{wikramaratna2010effects}. In theory, strong cross-reaction will involve three or more strains in the epidemic phase \cite{dutta2022covid, kamo2002effect, andreasen1997dynamics}. Boni (2004) further observed that the cross-reactivity influences the epidemic period \cite{boni2004influenza}. Medical interventions such as vaccines may change this situation \cite{martcheva2008vaccine}, but we assume in this paper that the model is constructed without any treatment. The natural mechanism of cross-reactivity manifests in two ways, namely ``reduces susceptibility" and ``reduced transmission" \cite{kucharski2016capturing, martcheva2008vaccine}. Reduced susceptibility refers to the hosts that have partial immunity because they have been infected by another strain, while reduced transmission means that they are less likely to transmit the specific strain to others \cite{kucharski2016capturing}.

\subsection*{Numerical framework of history-based multi-strain transmission}
We start our modification of the model \eqref{viggooriginal} by writing out as three equations for forces of infection, infected class, and susceptible classes. These make explicit the possible seasonality and use an indicator function $\mathbb{I}_{\{\circ{}\}}$ to avoid treating the empty set in a separate equation. This gives a new system as follow:
\begin{subequations} 
    \begin{align}
    \Lambda^i & ={\beta_i}(t) \sum_{\mathcal{J}\subseteq{\mathcal{K}\setminus\{i\}}}^{}I_\mathcal{J}^i
    \label{setfoi} ,\\
        \dot I_\mathcal{J}^i &= \sigma_\mathcal{J}^i \Lambda^i S_\mathcal{J} - (\mu+\gamma)I_\mathcal{J}^i \label{setIdot} ,\\
        \dot S_\mathcal{J} &= b(t)\mathbb{I}_{\{\mathcal{J}= \varnothing\}  } + \sum_{j \in \mathcal{J}}^{} \gamma I_{\mathcal{J}\setminus \{j\}}^j - \sum_{i\notin\mathcal{J}}^{} \sigma_\mathcal{J}^i \Lambda^i S_\mathcal{J} - \mu S_\mathcal{J}\label{setSdot} .
    \end{align}
        \label{setSIeqns}
\end{subequations}
The framework is constructed using set notation \cite{andreasen1997dynamics}, which makes it comprehensible and adaptable to e.g. dynamical systems analysis by hand. However, when dealing with numerous categories and various strain types, direct numerical simulation tends to perform slowly due to extensive looping processes required to generate the outcomes. 

We now start to represent the entire system in matrix format, focusing on enhancing efficiency. To address this goal, the system will be restructured and expressed as a matrix formulation. Let $i,j,...\in[n]=\mathcal{K}$ denote the type of strains where $n$ is the number of strains and we use the notation $\{n\}$ for the set of integers from 1 to $n$, $\{1,...,n\}$. Let $a,b,...\in[2^n]$ denote the possible infection histories where $2^n$ is the number of distinct subsets of $\mathcal{K}$. The transforms from sets to the integers,  $\phi:P(\mathcal{K}) \rightarrow [2^n]$ and from integers to sets via its inverse, $\phi^{-1}:[2^n] \rightarrow P(\mathcal{K})$. Logical equivalence between current strain being $i$ and the subset of history strain excluding current strain $i \notin \mathcal{J}$ will be guaranteed by the following proposition:
\begin{prop}
\label{prop1}
\text{Let} $i,j,...\in \mathcal{K}$ and $\mathcal{J} \subseteq \mathcal{K}$, then $i \notin \mathcal{J}$ \text{if and only if} $\mathcal{J} \subseteq \mathcal{K}\setminus \{i\}$.
\label{proposition_1}
\end{prop}

\begin{proof}
By the definition of set subtraction, $\mathcal{K} \setminus \{i\}$ contains all elements of $\mathcal{K}$ \text{except} $i$. Let $x \in \mathcal{J}$. Since $\mathcal{J} \subseteq \mathcal{K}  \text{ thus }  x \in \mathcal{K}$. Because $i \notin \mathcal{J} \text{ then } x \neq i$. Therefore, $x \in \mathcal{K} \text{ and } x \neq i, \text{ implies } x \in \mathcal{K} \setminus \{i\}$. Since $x$ is arbitrary, it follows that $\mathcal{J} \subseteq \mathcal{K} \setminus \{i\}. \text{ Thus } i \notin \mathcal{J} \Rightarrow \mathcal{J} \subseteq \mathcal{K} \setminus \{i\}$. By contradiction, it will be shown that if $\mathcal{J} \subseteq \mathcal{K} \setminus \{i\}$ then $i \notin \mathcal{J}$. Suppose that $i \in \mathcal{J}$. Then $i \text{ must also be in } \mathcal{K} \setminus \{i\} \text{ because } \mathcal{J} \subseteq \mathcal{K} \setminus \{i\}$. However, $\mathcal{K} \setminus \{i\}$ means that $i \notin \mathcal{K} \setminus \{i\}$ which contradicts the assumption that $i \in \mathcal{J}$. Therefore, $i \notin \mathcal{J}$. Thus $\mathcal{J} \subseteq \mathcal{K} \setminus \{i\} \Rightarrow i \notin \mathcal{J}$.
\end{proof}

This proposition will let us make a set manipulation that gives clarity about historical strain that excluding the current strain. Next, the dynamical variables in the system \eqref{setIdot} and \eqref{setSdot} can be written as a vector and matrix as follows:
\begin{align}
     \mathbf{S} &=(S_{a})_{a \in [2^n]} ; \\
    \boldsymbol{I} &=[I_{a.i}]_{a \in [2^n],i \in [n]} ;
\end{align}
\noindent
where $\mathbf{S}$ is a vector of length $2^n$ and  $\boldsymbol{I}$ is a matrix of size $2^n \times n$ respectively. The force of infection with strain $i$, $\Lambda^i$  from \eqref{setfoi} can be written as 
\begin{equation}
\begin{split}
    \Lambda^i(t) &=\beta^i \sum_{\mathcal{J} \subseteq \mathcal{K} \setminus \{i\}}^{} I(t)_\mathcal{J}^i,\\
&= \left(\Lambda^i(t)\right)_{i \in [n]} .
\end{split}
\label{matrixfoi}
\end{equation}
We therefore define a force of infection vector $\mathbf{\Lambda}(t)$ is the element-wise multiplication (Hadamard product) of vector of transmission rates $\boldsymbol{\beta}$ and a vector $\mathbf{J}(t) = (J_i)_{i\in [n]}$, $\mathbf{\Lambda}(t) = \boldsymbol{\beta}(t) \circ \mathbf{J}(t)$ with $\left({J}_i\right)_{i \in [n]}=\sum_{\mathcal{J} \in \mathcal{K} \setminus \{i\}}^{} I_\mathcal{J}^i$ describing the number of all individuals currently infectious with strain $i$ no matter what past history strain they have. Let $\boldsymbol{M}=[\boldsymbol{M}_{\phi(\mathcal{J}),i}]$ be $2^n \times n$ matrix, such that 
\begin{equation}
\boldsymbol{M}_{\phi(\mathcal{J}),i} = 
        \begin{cases} 
          1 &  \text{if } \mathcal{J} \subseteq \mathcal{K}\setminus 
          \{i\}\\
          0 & \text{otherwise.}
       \end{cases}
       \label{M_matrix}
\end{equation} 
This $\boldsymbol{M}$ represents the indicator matrix where its elements consist of 0 and 1 based on the condition in Eq.~\eqref{M_matrix}. We now derive the details of $\left({J}_i\right)_{i \in [n]}$ whereby it can be represent as the sum of dot product of $\boldsymbol{M} \text{ and } \boldsymbol{I}$ with associated strain in the map $\phi$: 
\begin{equation}
\begin{split}
      J_i&=\sum_{a \in [2^n]}{} [\boldsymbol{M}]_{ai} [\boldsymbol{I}]_{ai} = [\boldsymbol{M}^{\mathsf{T}} \boldsymbol{I}]_{ii}\\
        \Rightarrow \mathbf{J}&=\mathrm{diag}(\boldsymbol{M}^{\mathsf{T}} \boldsymbol{I}) . \label{J_matrix}\\
\end{split}
\end{equation}
       
The cross reactivity parameters $\sigma_\mathcal{J}^{i}$ can be represented using a matrix
$\boldsymbol{\sigma} = [\sigma_{a,i}]$ where strain $i$ interacts with strain history indexed by $a \in [2^n]$. This lets us write equation \eqref{setIdot} for the elements of $\dot{\boldsymbol{I}}$ in matrix notation:
\begin{equation}   
    \begin{aligned}
    \dot{\boldsymbol{I}}=\boldsymbol{\sigma} \circ (\mathbf{1}_{2^n} \mathbf{\Lambda}^{\mathsf{T}} ) \circ (\mathbf{S} \mathbf{1}_n^{\mathsf{T}}) - (\gamma + \mu) \boldsymbol{I} ,
    \end{aligned}
\end{equation}
\noindent
where $\mathbf{1}_\ell$ is a length-$\ell$ vector whose elements are all equal to $1$ and $(\boldsymbol{\sigma} \circ (\mathbf{1}_{2^n} \mathbf{\Lambda}^{\mathsf{T}} ) \circ (\mathbf{S} \mathbf{1}_n^{\mathsf{T}}))$ is the incidence of a new cases in population (class) that comes from the susceptible population that has been infected by $i$ strain and potentially experienced cross-immunity. We now turn to \eqref{setSdot}, which can be written in the form
\begin{equation}
    \begin{aligned}
    \dot{\boldsymbol{S}}_\mathcal{J} = b(t)\mathbf{u}_{\phi(\varnothing)} +
    \mathbf{V}_\gamma -\mathbf{V}_\sigma - \mu \mathbf{S},
    \end{aligned}
\end{equation}
\noindent
where $\mathbf{u}_a$ is a vector of all $0$ except for the $a$-th element, which is $1$.
We now need to consider the middle two terms that we have used $\mathbf{V}$'s to stand in for above. The vector $\mathbf{V}_\gamma$ represents the sum of the terms associated with recovery while $\mathbf{V}_\sigma$ is the sum terms associated with infection. To transform the system from set to integer, we will also need an invertible map $\varphi$ between $2^n \times n$ matrices and length-$(n 2^n)$ vectors such that
\begin{align*}
    \varphi : [2^n] \times [n] \leftrightarrow [n 2^n]
\end{align*}

Let $\sum_{i \in \mathcal{J}}{} I_{\mathcal{J} \setminus\{i \}}$ describes the class infected by strain $i$ after they got infected by all strains in $\mathcal{J} \setminus\{i\}$. This can be formulated as:
\begin{equation}
    \begin{aligned}
        \sum_{i \in \mathcal{J}}{} I_{\mathcal{J} \setminus \{i \} }&=
        \sum_{i \in \mathcal{J}{}} \sum_{\mathcal{I}} \boldsymbol{I}_{\{ \mathcal{I} = \mathcal{J}\setminus 
          \{i\} \}} I_\mathcal{I}^{i} \\
          &= \sum_{i,\mathcal{I}{}} \boldsymbol{I}_{\{ i \in \mathcal{J}  \text{ and }  \mathcal{I} = \mathcal{J}\setminus 
          \{i\} \} } I_\mathcal{I}^{i}\\
          & = \sum_{i,\mathcal{I}{}} q_{\mathcal{J},\mathcal{I},i} I_\mathcal{I}^{i} \, ,
     \end{aligned}
\end{equation}       
where
\begin{align}
 q_{\mathcal{J},\mathcal{I},i} = 
        \begin{cases} 
          1 &  \text{if }  i \in \mathcal{J}  \text{ and }  \mathcal{I} = \mathcal{J}\setminus 
          \{i\} \\
          0 & \text{otherwise.}
       \end{cases}
\end{align}
This 3-dimensional object serves as an indicator, comprising the current strain $i$, history of strains $\mathcal{J}$, and the subsequent infection phase with different strains originating from recovered individuals, specifically denoted by $\mathcal{I} \setminus \{i\}$. In this context, the notation $\mathcal{I} \setminus \{i\}$ is used to indicate that, during the second, third, or any subsequent cycle of infection, it is presumed that the infection is caused by different strains than those previously involved.

Then let
\begin{align*}
    \boldsymbol{Q} = [Q_{a,A} ]_{a\in [2^n], A \in [n2^n]} ,
\end{align*}
where
\begin{align*}
     Q_{\phi(\mathcal{J}),\varphi(\phi(\mathcal{I}),i)} = q_{\mathcal{J},\mathcal{I},i} .
\end{align*} 
We will then be able to write
$$
\mathbf{v}(\boldsymbol{I}) = (\mathrm{v}_A(\boldsymbol{I})), \qquad \text{ where } \quad
    \mathrm{v}_{\varphi(\phi(\mathcal{I}),i)}(\boldsymbol{I}) = I_\mathcal{I}^{i} \, , 
$$
and also
$$
    \mathbf{V}_{\gamma} = \gamma \boldsymbol{Q} \mathbf{v}(\boldsymbol{I}). 
$$
Note that from Proposition \ref{proposition_1}, $i \notin \mathcal{J}$ if and only if $\mathcal{J} \subseteq \mathcal{K}\setminus \{i\}$. 
Mathematically, therefore, $\mathbf{V}_\sigma$ is formulated as follows:
\begin{center}
    $\mathbf{V}_\sigma = \mathrm{diag}\left( \boldsymbol{M}^\mathsf{T}({\boldsymbol{\sigma}} \circ (\mathbf{1}_{2^n} \cdot \boldsymbol{\Lambda}^\mathsf{T}) \circ (\mathbf{S} \mathbf{1}_n^ \mathsf{T}))\right)$
\end{center}

In summary, the system presented in Eq. \eqref{setSIeqns} can been converted into a matrix format, which simplifies numerical computations. This transformation enhances the efficacy of simulations because the conversion addresses the issue of looping indices in the original set of equations. This is achieved through the use of indicator matrices, which are composed of binary values, specifically 0s and 1s, thereby streamlining the process. The new numerical framework for the multi-strain model is presented as follow:  
\begin{subequations} 
\begin{align}
    \bm{\Lambda}(t) & = \bm{\beta}(t) \circ \mathrm{diag}\left(\boldsymbol{M}^{\mathsf{T}} \boldsymbol{I}\right),
    \label{matfoi}\\
        \boldsymbol{\dot I} & =\boldsymbol{\sigma} \circ \left(\mathbf{1}_{2^n} \mathbf{\Lambda}^{\mathsf{T}} \right) \circ \left(\mathbf{S} \mathbf{1}_n^{\mathsf{T}}\right) - (\gamma + \mu) \boldsymbol{I}, \label{matIdot}\\
    \dot{\mathbf{S}} &= b(t)\mathbf{u}_{\phi^{-1}(\varnothing)} + \gamma \bm{Q} \mathbf{v}(\boldsymbol{I}) - \mathrm{diag} \left( \bm{M}^\mathsf{T}\left({\boldsymbol{\sigma}} \circ \left(\mathbf{1}_{2^n} \boldsymbol{\Lambda}^\mathsf{T}\right) \circ \left(\mathbf{S} \mathbf{1}_n^\mathsf{T}\right)\right)\right) - \mu \mathbf{S}.
\end{align}
        \label{matSIeqns}
\end{subequations} 
It is important to observe that working with this matrix system requires the use of matrix operations to ensure that every dimension of the vectors aligns appropriately for the operation to be successfully executed, meaning it becomes necessary to reshape the dimensions of the matrix to facilitate the numerical simulation without altering the perception of the matrix's original size. 

Table \ref{tab:operation} is presented to document the matrix operations involved in converting a set equation into an integer-based matrix system. This conversion process is simulated using the \texttt{Python} \cite{Python} programming language \texttt{version 3.10}, specifically employing the \texttt{NumPy} package.

\begin{table}[H]
\centering
\caption{Summary of Matrix Operation in \texttt{Python} Implementation}
\begin{tabular}{@{} l l l @{}}
\hline
Operation & Definition & \texttt{Python} Implementation \\
\hline
Matrix product & $\mathbf{AB} = \sum_k A_{ik}B_{kj}$ & $\mathbf{A} @ \mathbf{B}$ \\
Hadamard (element-wise) product & $[\mathbf{A \circ B}]_{ij} = A_{ij} B_{ij}$ & $\mathbf{A} * \mathbf{B}$ \\
Outer product & $[\mathbf{A \otimes B}]_{ij} = A_{i} B_{j}$ & $\text{np.outer}(\mathbf{A,B})$ \\
Matrix transpose & $\mathbf{A}^{\mathsf{T}}$ & $\mathbf{A}.\text{T}$ \\
Set to Integer transformation & $\phi: P(\mathcal{K}) \rightarrow [2^n]$ & \texttt{set2int(a)} \\
Integer to Set transformation & $\phi^{-1}:[2^n] \rightarrow P(\mathcal{K})$ & \texttt{int2set(a)} \\
Set difference & $\mathbf{A}\setminus\{i\}$ & $\mathbf{A}$\texttt{.difference}$(\{i\})$ \\
\hline
\multicolumn{3}{@{}l@{}}{%
    \parbox{\linewidth}{\smallskip
    {\small All code is available at \url{https://github.com/fadilah-ilahi/multi-strain-framework}.}}
} \\
\end{tabular}
\label{tab:operation}
\end{table}

Since this framework's goal is to improve computational efficiency in terms of numerical simulation runtime, these matrix operations play a crucial role.

The upcoming subsection will provide an explanatory illustration on how this framework can be applied to the cases of disease transmissions with multiple pathogens like those seen in COVID-19 \cite{dutta2022covid,kim2022model, wang2022multiscale, li2022evaluating} that involve strain replacement \cite{lythgoe2023lineage, martcheva2008vaccine, majeed2022variant, Murall2014RevisingReplacement}.

\subsection*{Strain replacement}

We demonstrate here the utility of our approach by performing a numerical analysis of strain replacement in situations where the disease involves multiple serotypes. Following Martcheva et al. \cite{martcheva2008vaccine}, strain replacement is defined as the substitution of dominant pathogen strains over time by other strains. This is a dynamic process influenced by complex interactions at both within host and population levels. The paper explores the mechanisms of pathogen strain replacement such as superinfection and co-infection. It also emphasises that fully protecting treatment such as perfect vaccine against all strains may lead to strain replacement \cite{majeed2022variant} if the strains differ in their competitive abilities within hosts. Studies by Gog and Grenfell (2002) \cite{gog2002dynamics} and Kucharski et.al (2012) \cite{kucharski2012original, kucharski2016capturing} regarding strain replacement particularly involving antigenic \cite{mulberry2023pneumococcal} and genetic mutation \cite{kim2022model, li2022evaluating}, have already shown how immunity plays a crucial role in shaping strain dynamics. The model predicts that strains shift to the most closely related strains that represent movement to neighbouring points in strain space \cite{kucharski2016capturing, andreasen1997dynamics}. Physical distance in strain space corresponds to antigenic \cite{kim2022model, Truscott2012EssentialInfluenza} similarity and hence cross-immunity strength. 

Recall that the strain replacement is accommodated using cross-reactivity \cite{Castillo-Chavez1989EpidemiologicalCross-immunity, Woodall2014PartialEpidemiology, kamo2002effect, Murall2014RevisingReplacement} parameter $\sigma_\mathcal{I}^i$ that later becomes cross-reactivity matrix $\boldsymbol{\sigma}$. Let history strain $\mathcal{I}$ represent prior exposure to strain $i$ where the $n \times 2^n$ a value of 1 indicates complete immune escape while 0 value means full cross-immunity. Adapting the concept of linear strain space of Gog and Grenfell in \cite{gog2002dynamics}, there are $\varepsilon^{\mathrm{min}\{|i-j|\}_{j \in \mathcal{I}}}$ that describes antigenic space to control immunity decays as strains diverge where ${\mathrm{min}\{|i-j|\}_{j \in \mathcal{I}}}$ is a minimum antigenic distance or shortest distance between strains. In addition, $\sigma^i_\mathcal{I}$ can be written as
\begin{equation}
    \boldsymbol{\sigma} = 
        \begin{cases} 
          1 &  \text{if } \quad \mathcal{I} = \varnothing\\
          0 & \text{if } \quad i \in \mathcal{I}\\
          1-\varepsilon^{\mathrm{min}\{|i-j|\}_{j \in \mathcal{I}}} & \text{otherwise.}
       \end{cases}
\end{equation}
Note that $\varepsilon \in (0,1)$ determines the rate of immunity retention across antigenically distinct strains. By using the $\boldsymbol{\sigma}$ above, the model is assumed to reduce immunity between adjacent strains by a constant multiplicative factor $\varepsilon$ which is justified by Khoury et al.\ \cite{khoury2021neutralizing}, Wang et al.\ \cite{wang2021antibody}, and Wang et al.\ \cite{wang2025optimizing}. 
It also implies that immunity to a given strain provides full protection when $\boldsymbol{\sigma}=1$. The use of the minimum antigenic distance ensures the capture of gradual immune escape and facilitates a smooth transition between closely related variants.

Suppose that we have two strains circulating in the population. To estimate the proportions of strain 1 (and thereby strain 2), this study applies Gaussian Process Classification (GPC) model of \cite{lythgoe2023lineage}, which we explain here. Let $(t_i,y_i)_{i=1}^n$ denote observations where the $t_i \in [t_0, t_n]$ are the times and the $y_i$ are the binary function to label each variants, so that in mathematical formula it also can be written as:
\begin{equation*}
    y_i = \begin{cases} 
          1 &  \text{if } i \text{ is strain 1;} \\
          0 & \text{if } i \text{ is strain 2.}\\
       \end{cases}
\end{equation*} 
This process of estimating the proportion of all variants in population $\pi(t)$ is using log odds with $z_1 \text{ and } z_2$ are strain 1 and strain 2 respectively, so that:
\begin{equation}
    f(t) = \text{log} \left( \frac{\pi(t)}{1-\pi(t)}\right) = \text{log} \left( \frac{z_1(t)}{z_2(t)}\right)
    \label{log-growth}
\end{equation}
where 
\begin{equation}
    \pi(t) = \frac{z_1(t)}{z_1(t) + z_2(t)}
    \label{proportion}
\end{equation}
is the proportion of variant 1 in the population. Our GP kernel \cite{Martin2016BayesianFramework} uses the Radial Basis Function (RBF):
\begin{equation*}
    k(t,t')=\sigma^2 \text{exp}\left(-\frac{(t-t')^2}{2\ell{}^2}\right).
\end{equation*}
Note that $z_1(t) \text { and } z_2(t)$ follow exponential growth. The relative growth advantage for competing variant is calculated by derivative of the proportion $f(t)$ which can be written as
\begin{equation}
    \delta r(t) = \mathrm{\frac{d}{dt}} \text{log}\left(\frac{\pi(t)}{1-\pi(t)}\right) 
    \label{delta_r(t)}
\end{equation}
\noindent
This relative growth advantage represents how much faster one strain grows relative to another. Let say strain 1 is growing faster compared to strain 2, so   $\delta r(t) > 0$, it means that strain 1 is increasing relative to strain 2.

\section*{Results and Discussion}

Our analysis examined the two-strain epidemic model associated with COVID-19 data in Indonesia. The data collected from GISAID (via \url{https://www.epicov.org/}) capture the circulation of Delta and Omicron variants \cite{haykal2024dynamics} as the two dominant strains in the population between July 2021 and January 2023. During this period, Delta and Omicron variants circulated, where eventually the Omicron replaced the Delta variant over a time \cite{prasetyoputri2022dynamics}. This phenomenon is called strain replacement or a ``sweep". This empirical pattern provides a biologically relevant test case for evaluating the proposed framework. Next, we conduct the same analysis of two-strain model using matrix framework to confirm that the model and the real data plot have a similar pattern. Having established qualitative agreement in the low-dimensional case, we then extended the framework to higher-dimensional systems involving 5 and 11 circulating strains to examine its generalisability and dynamic behaviour under increased strain diversity.

Because the number of compartments capturing possible infection histories grows exponentially with the number of strains, higher-dimensional simulations provide not only biological insight but also a test of computational scalability. In the last subsection of the result, we therefore additionally compared the runtime performance of the proposed matrix-based implementation against a direct set-indexed formulation of the same model. This comparison quantified the computational gains achieved through vectorised matrix operations while preserving the underlying epidemiological structure.

\subsection*{Analysis of SARS-CoV-2 variants in strain replacement}
We begin this section by describing the result from the Indonesian SARS-CoV-2 (COVID-19) data specifically in terms of strain replacement between Delta and Omicron. The first result is provided in Fig.~\ref{fig:weekly-cases} showing the number  of Delta and Omicron increased in the population. During the period around January 2021, the variant curves overlapped, indicating that Delta was decreasing while Omicron began to increase. Using the system in Eq.~(\ref{matSIeqns}) we simulate the two-strain model that involved strain replacement with $\beta_2 > \beta_1$ (Fig.~\ref{fig:two-strain-model}). The model parameters here $\beta_1 = 0.32, \beta_2 = 0.5, \sigma = 0.8, \gamma = 1/7$ and $\mu = 1/(70 \cdot 365)$ were selected to reproduce the observed Delta–Omicron replacement pattern in the empirical data. In more specific way, we set $\beta_2 > \beta_1$ reflects higher transmissibility of Omicron relative to Delta \cite{haykal2024dynamics, prasetyoputri2022dynamics} and the cross-reactivity value $\sigma = 0.8$ represents partial immune escape that consistent with the partial but incomplete cross-protection reported between these two variants. This is also confirmed by two-strain model of Eq.~\eqref{matSIeqns} in the Fig.~\ref{fig:two-strain-model} regarding the replacement process between strain 1 and strain 2. The vertical dash line noted the first time that replacement has occurred in the model. In two-strain model (right panel of Fig.~\ref{fig:data-vs-model}), we assume a population size $N = 1000$ is chosen for illustrative purposes and does not represent the true Indonesian population size. The goal of Fig.~\ref{fig:two-strain-model} is not to reproduce the exact scale of the epidemic, but to demonstrate that the qualitative dynamics of strain replacement are consistent with the pattern observed in the empirical data.

\begin{figure}[H]
        \captionsetup[subfigure]{justification=centering}
        \centering
        \begin{subfigure}[b]{0.475\textwidth}
            \centering
            \includegraphics[width=\textwidth]{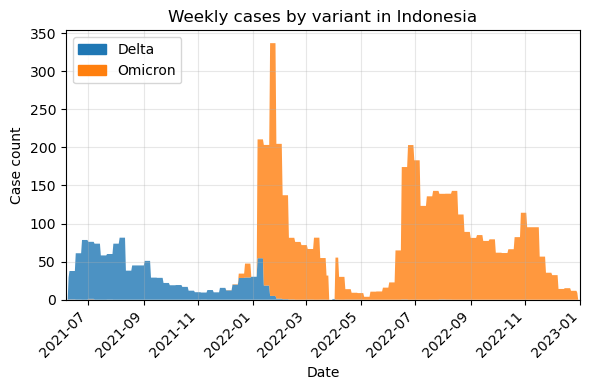}
            \caption{\small}
            \label{fig:weekly-cases}
        \end{subfigure}
        \hfill
        \begin{subfigure}[b]{0.475\textwidth}
            \centering
            \includegraphics[width=\textwidth]{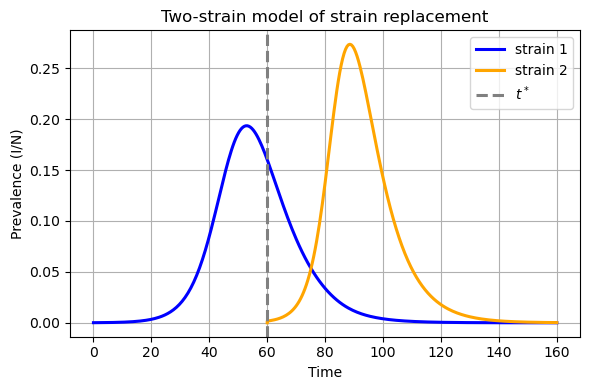}
            \caption{\small}
            \label{fig:two-strain-model}
        \end{subfigure}

        \begin{subfigure}[b]{0.475\textwidth}
            \centering
            \includegraphics[width=\textwidth]{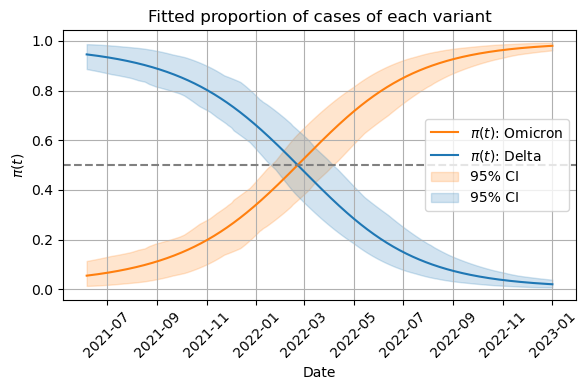}
            \caption{\small}
            \label{fig:prob-growth}
        \end{subfigure}
        \hfill
        \begin{subfigure}[b]{0.475\textwidth}
            \centering
            \includegraphics[width=\textwidth]{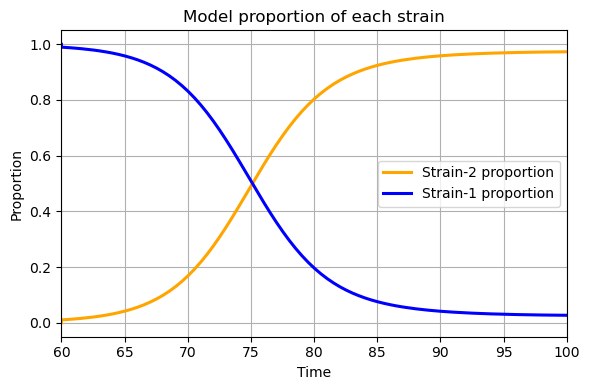}
            \caption{\small}
            \label{fig:prop-two-strains}
        \end{subfigure}

        \begin{subfigure}[b]{0.475\textwidth}
            \centering
            \includegraphics[width=\textwidth]{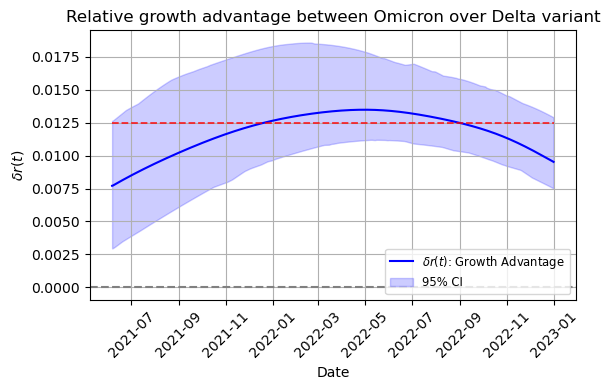}
            \caption{\small}
            \label{fig:growth-advantage}
        \end{subfigure}
        \hfill
        \begin{subfigure}[b]{0.475\textwidth}
            \centering
            \includegraphics[width=\textwidth]{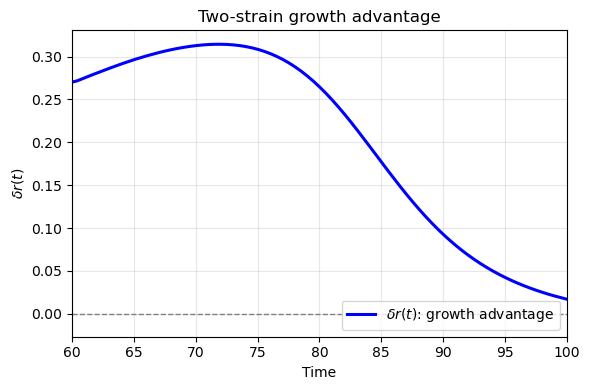}
            \caption{\small}
            \label{fig:growth-advantage-model}
        \end{subfigure}
        
        \caption[]
        {\small \textbf{Data vs model of two-strains dynamics over time.} (a) Number of prevalence of Covid-19 in Indonesia with Delta (blue) and Omicron (orange) circulated in the population from July 2021 to January 2023. (b) Two-strain model dynamic of strain replacement with $\beta_1 = 0.32$, $\beta_2 = 0.5$, $\gamma = 1/7$, $\mu = 1/(70 \cdot 365)$, $N = 1000, \text{ and } \sigma=0.8$. Assumed that strain 1 is replace by strain 2 with $\beta_2 > \beta_1$. (c) Fitted proportion cases with Delta and Omicron variant based on Covid-19 data in Figure (a). (d) Proportion of two variants in Figure (b). (e) Relative growth advantage between Omicron over Delta variant with 95\% CI. Note that the red dash in $\delta(t)=0.0125$ represents the relative growth advantage where Omicron's growth was stronger or weaker than the average. (f) Two-strain growth advantage from the model where it reflects the same trend as the Delta and Omicron data.} 
        \label{fig:data-vs-model}
    \end{figure}

Next, we statistically fit the proportion of Delta over Omicron and Omicron over Delta using Gaussian Process Classification (GPC) \cite{Martin2016BayesianFramework} with 95\% of confidence interval described in the methods. This analysis produces the expected trajectory, showing that Omicron swept the Delta variant following a sinusoidal pattern (Fig.~\ref{fig:prob-growth}). It can be interpreted that from July 2021, Delta variants dominated. However, from December 2021 onward, Omicron increased exponentially and eventually stable by the end of 2022 while Delta declined toward extinction. The cross-over point between Delta and Omicron occurred in April 2022 at which point each variant accounted for 50\% of infections. In this case, the time from emergence around November 2021 to near-complete dominance took roughly 4–5 months.

Examining the model simulation more closely in Fig.~\ref{fig:prop-two-strains}, we see similar proportions to real data, when one strain comes to dominates the other. Using the same parameter values, this ultimately demonstrates that the framework captures key dynamics and reproduces the observed data pattern \cite{souto2019multiple}. Initially, strain 1 (represented by blue curve) is the dominant variant. However, as time progresses from day 60 to 100, a complete shift in dominance occurs with strain 2 emerges and increases rapidly in prevalence.

The rate of strain replacement is captured in the growth advantage $\delta r$ in Fig.~\ref{fig:growth-advantage} and Fig.~\ref{fig:growth-advantage-model}. These plots explain how fast each strain grows relative to the others. This real-world point estimate of $\delta r(t)$ peaks during the initial phase of strain replacement before declining. The red horizontal dashed line at $\delta r = 0.0125$ however, sits within the 95\% confidence interval suggesting that the date are constant with the advantage of Omicron being relatively constant.  The growth advantage in the model is controlled by factor of $\boldsymbol\sigma$ that represents cross-reactivity giving partial immunity in susceptible population. In this case, cross-reactivity values give the population 20\% partial immunity. The fact that empirical data of COVID-19 in Delta and Omicron replacements validates model structure that we constructed confirms that competitive interaction between two-strains is a dynamic process with a pattern that is seen in both theoretical model and the real-world variants data.

\subsection*{Replacement with a larger number of strains}

Building on the validation of the two-strain model from Fig.\ref{fig:data-vs-model} with real data, now we consider the framework of multi-strain model to examine the behaviour under a larger number of strains circulating in population. The result in Fig.\ref{fig:two-strain-model} confirms that the proposed matrix-based system can reproduce the epidemiological patterns as present in real life COVID-19 data for Delta and Omicron replacement observed in Indonesia. It demonstrates accuracy in both model's consistency with data and the robustness of the numerical implementation. Note that working with a larger number of strains will produce larger matrices in the system. To test its generalisation, we applied the same system in eq.\eqref{matSIeqns} to simulate 5-strain and 11-strain interactions with varying transmission rate vector $\boldsymbol{\beta} = (\beta_i)$ with a fixed cross-reactivity decay $\varepsilon=0.2$.

In general, the dynamics of the 5-strain model follow the 2-strain replacement model that had been discussed in previous subsection. Different strains invade the population after the preceding one goes extinct. It happens over and over again depending on how many strains are circulating in the population. As shown in Fig.~\ref{fig:simulation-5strains}, the numerical framework successfully generates the infection waves depends on the number of strain circulated where strain dominance is ordering by basic reproduction number $\mathcal{R}_0$. The left column of simulations represents the number of infection by multiple strains with three different settings based on the transmission vector $\boldsymbol{\beta}$: (1) gradually decrease $\beta_1 < \beta_2 < \ldots{} < \beta_n$, (2) constant $\beta_1 = \beta_2 = \ldots{} = \beta_n$, and (3) gradually increase $\beta_1 > \beta_2 > \ldots{} > \beta_n$. This transmission vector also corresponds to the basic reproduction number $\mathcal{R}_0$ that calculated by $\mathcal{R}_0 = \boldsymbol{\beta} \oslash \left(\boldsymbol{\gamma} + \mu \right)$ as an element-wise operation. Note that $\boldsymbol{\gamma}$ vector is the Hadamard product of recovered parameter $\gamma$ and the number of strains $n$, while $\boldsymbol{\beta}$ is a $2^n \times n$ transmission matrix.

The first strain in Fig.~\ref{fig:5strains-decrease} reaches its maximum around $t=30$ with approximately almost a half number of infection from the total of population and later on appear at longer intervals with lower magnitudes \cite{allen2008mathematical} for the following strain. This occurs as a consequence of the cross-reactivity mechanism, where it can be reduced susceptibility or reduce the ability of transmission \cite{andreasen1997dynamics}. Reduced susceptibility refers to the hosts that build part immunity in their body because they have infected the specific strain, while reduced transmission ability
means that they are less likely to transmit the specific strain to others \cite{andreasen1997dynamics}. This also reflects that basic reproduction number decreases as new strains fail to sustain for larger outbreak. Each peak represents different strain invaded that depends on the rate of transmission. For increasing value of $\boldsymbol{\beta}$, the $\mathcal{R}_0$ will also gradually increasing values. 

In the same manner, it occurs in constant (Fig.~\ref{fig:5strains-constant}) and decreasing pattern (Fig.~\ref{fig:5strains-decrease}). This pattern has been explained by Sayyar and R{\"o}st in 2022 \cite{sayyar2022epidemic} about generated pattern of epidemic under different assumption of cross-immunity. Under certain assumption, the number of infections of the first strain decreases due to recovery and the new strains will sequentially replace the previous ones in dominance. This process repeats until the last strain and eventually becomes stable \cite{gog2002dynamics}.

\begin{figure}[H]
        \captionsetup[subfigure]{justification=centering}
        \centering
        \begin{subfigure}[b]{0.475\textwidth}
            \centering
            \includegraphics[width=\textwidth]{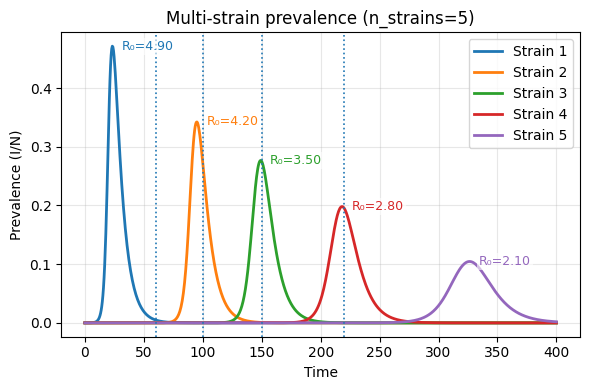}
            \caption{\small}
            \label{fig:5strains-decrease}
        \end{subfigure}
        \hfill
        \begin{subfigure}[b]{0.475\textwidth}
            \centering
            \includegraphics[width=\textwidth]{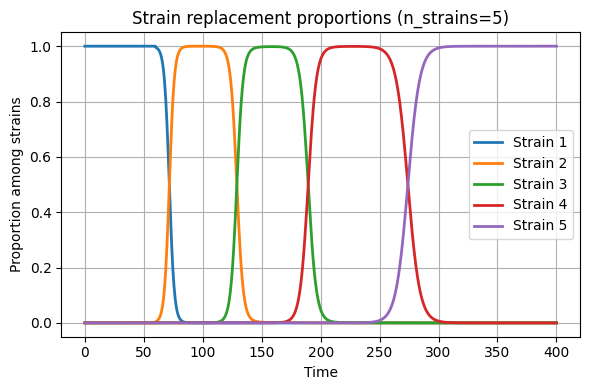}
            \caption{\small}
            \label{fig:prop-5strains-decrease}
        \end{subfigure}

        \begin{subfigure}[b]{0.475\textwidth}
            \centering
            \includegraphics[width=\textwidth]{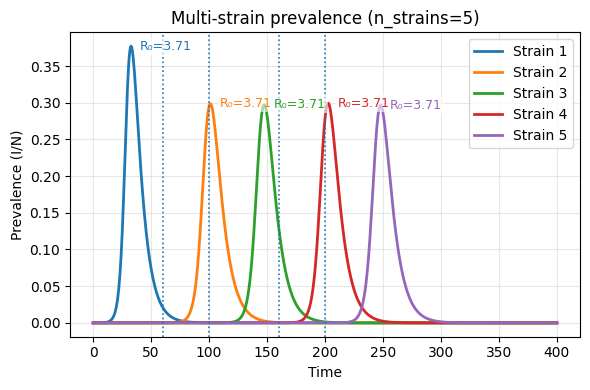}
            \caption{\small}
            \label{fig:5strains-constant}
        \end{subfigure}
        \hfill
        \begin{subfigure}[b]{0.475\textwidth}
            \centering
            \includegraphics[width=\textwidth]{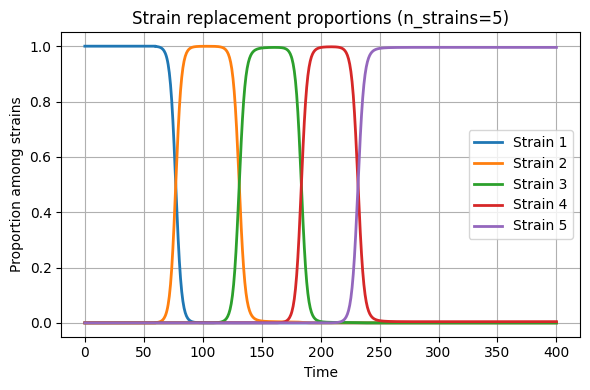}
            \caption{\small}
            \label{fig:prop-5strains-constant}
        \end{subfigure}

        \begin{subfigure}[b]{0.475\textwidth}
            \centering
            \includegraphics[width=\textwidth]{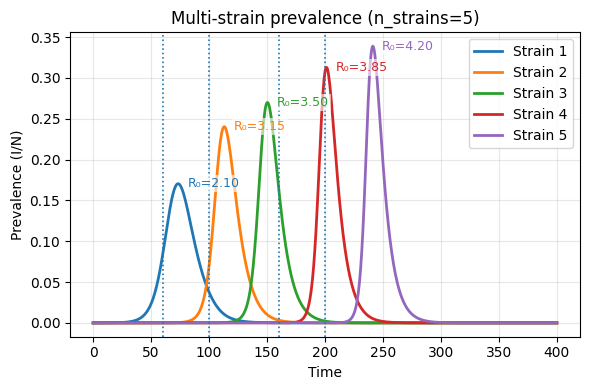}
            \caption{\small}
            \label{fig:5strains-increase}
        \end{subfigure}
        \hfill
        \begin{subfigure}[b]{0.475\textwidth}
            \centering
            \includegraphics[width=\textwidth]{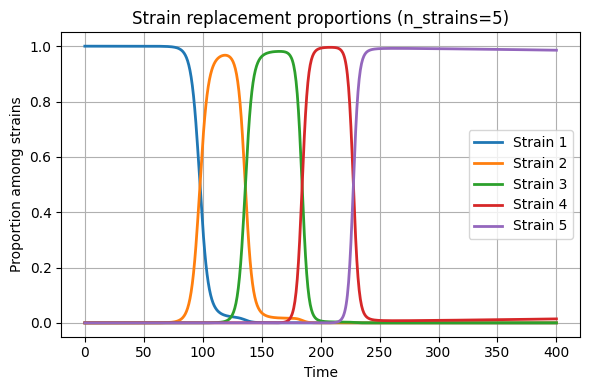}
            \caption{\small}
            \label{fig:prop-5strains-increase}
        \end{subfigure}
        
        \caption[]%
        {\small \textbf{Simulation for $5$-strain replacements.} Each peak represents the number of infection caused by each strain with $\varepsilon = 0.2$, $N = 10^5$, $\gamma = 1/7$, $\mu = 1/(70 \cdot 365)$. These figures are 5-strains replacement and proportion with (a) $\mathcal{R}_0$ is decreasing as $\beta$ is decreasing. (b) 5-strain proportion with $\beta$ decreasing. (c)  $R_0$ constant as $\beta$ is in the same constant value. (d) 5-strain constant proportion. (e) $R_0$ is increasing as $\boldsymbol{\beta}$ is increasing. (f) 5-strain proportion with increasing elements of $\beta$.} 
        \label{fig:simulation-5strains}
    \end{figure}

The strain proportion is qualitatively similar in three different settings of increasing, constant, and decreasing rate of transmission $\beta$. However, a more rapid infection also gives a rapid strain replacement (Fig.~\ref{fig:prop-5strains-increase}). Note that this replacement model does not reflect the severity of disease transmission, and just explains how fast the next strain will spread in population compared to the previous strain. This condition was confirmed with two-strain replacement model as in Fig.~\ref{fig:data-vs-model} where the later strain has a higher peak than the previous one, in that case, Delta vs Omicron. Further simulation about this is also conducted by Fig.~\ref{fig:5strains-increase} where strain $n+1$ spreads much faster than strain $n$ with the basic reproduction number gradually increases following the value of $\beta$ vector.

Finally, in the next figure (Fig.~\ref{fig:simulation-11strains}) we would like to generalise the model into the case of 11 strains circulating in the population. The pattern on both number of prevalence as well as strain proportion does not change so much from the 5-strain infection except it require longer timeline to reach its stability until zero infection condition. All code to produce all result figure is available at \url{https://github.com/fadilah-ilahi/multi-strain-framework}.

    \begin{figure}[H]
        \captionsetup[subfigure]{justification=centering}
        \centering
        \begin{subfigure}[b]{0.475\textwidth}
            \centering
            \includegraphics[width=\textwidth]{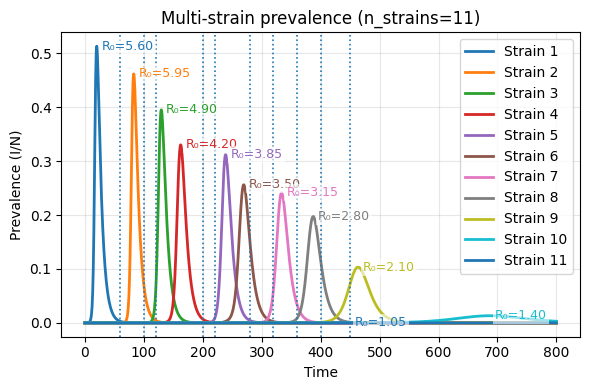}
            \caption{\small}
            \label{fig:11-strains}
        \end{subfigure}
        \hfill
        \begin{subfigure}[b]{0.475\textwidth}
            \centering
            \includegraphics[width=\textwidth]{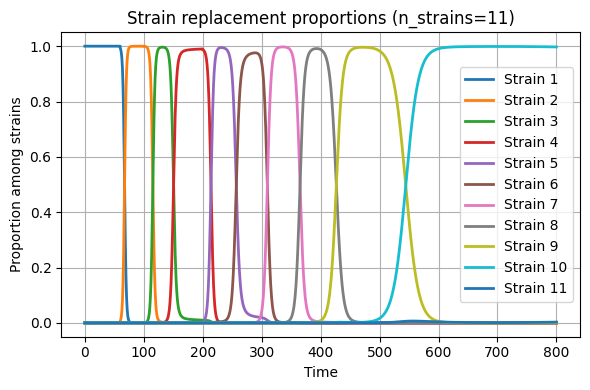}
            \caption{\small}
            \label{fig:prop-11-strains}
        \end{subfigure}
        
        \caption[]%
        {\small \textbf{Simulation for 11-strains replacements.} Each peak represents the number of infection caused by each strain with $\varepsilon = 0.2$, $N = 10^5$, $\gamma = 1/7$, $\mu = 1/(70 \cdot 365)$ with (a) $R_0$ of 11-strain dynamic with decreasing $\beta$, (b) 11-strain proportion with decreasing $\beta$.  } 
        \label{fig:simulation-11strains}
    \end{figure}

As shown in Fig.\ref{fig:simulation-11strains}, the numerical framework reproduces the sequential replacement behaviour outlined in Andreasen \cite{andreasen1997dynamics} which captures both interacting strains and their infection histories.

\subsection*{Model benchmark and computational performance}

The computational cost of history-dependent multi-strain epidemic models implemented in \texttt{Python 3.10} version. The model increases rapidly with the number of circulating strains due to the exponential growth of compartmental states as a consequence of tracking the history of prior strain infections. Based on Andreassen et al.\ \cite{andreasen1997dynamics, kucharski2016capturing}, the total number of compartments is given by $2^{n-1}(n+2)$ consisting of $2^n$ susceptible classes and $n2^{n-1}$ infected classes. As shown in Table~\ref{tab:benchmark}, the size of the system increases from only 8 compartments for $n=2$ strains to 13,312 states for $n=11$.

\begin{table}[H]
\centering
\caption{Summary of model benchmark runtime associated with number of strains}
\begin{tabular}{c r r r r}
\hline
\makecell{Number of\\strain $(n)$} &
\makecell{Total classes\\$2^{n-1}(n+2)$} &
\makecell{Set runtime\\(seconds)} &
\makecell{Matrix runtime\\(seconds)} &
Speedup \\
\hline
2 & 8 &$0.025 \pm 0.001$ & $0.010 \pm 0.003$ &  $2.47$\\
3 & 20 &$0.069 \pm 0.003$ & $0.011 \pm 0.001$ & $6.09$\\
4 & 48 &$0.175 \pm 0.010$ & $0.010 \pm 0.000$ &  $17.28$\\
5 & 112 &$0.419 \pm 0.011$ & $0.013 \pm 0.001$ & $31.68$\\
6 & 256 &$1.005 \pm 0.009$ & $0.013 \pm 0.000$ & $76.79$\\
7 & 576 &$2.498 \pm 0.294$ & $0.016 \pm 0.000$ & $153.57$\\
8 & 1,280 &$5.518 \pm 0.284$ & $0.022 \pm 0.001$ & $250.44$\\
9 & 2,816 &$12.496 \pm 0.204$ & $0.039 \pm 0.006$ & $322.50$\\
10 & 6,144 & $27.341 \pm 0.318$ & $0.052 \pm 0.000$ & $525.57$\\
11 & 13,312 &$60.103 \pm 0.345$ & $0.108 \pm 0.004$ & $557.95$\\
\hline
\end{tabular}
\label{tab:benchmark}
\end{table}

The set-based implementation used for benchmarking is a direct \texttt{Python} translation of system \eqref{setSIeqns}, in which each compartment $S_\mathcal{J}$ and $I^i_\mathcal{J}$ is indexed explicitly using \texttt{Python} set objects with some operations in Table \ref{tab:operation}. For each time step of the ODE solver, the right-hand side is evaluated by iterating over all subsets $\mathcal{J} \subseteq \mathcal{K}$ using nested loops, computing the force of infection and recovery terms through explicit enumeration of the power set. The matrix-based implementation, by contrast, pre-computes the indicator matrices $\boldsymbol{M} \text{ and } \boldsymbol{Q}$ prior to integration and evaluates the right-hand side entirely through \texttt{NumPy} matrix operations where the process eliminates all explicit iteration over strain subsets. Both implementations produce numerically identical output and use the same ODE solver \texttt{scipy.integrate.solve\_ivp} with the \texttt{RK45} method in the same settings and environment to ensure that the runtime comparison reflects only the difference in algebraic implementation strategy.

To evaluate computational efficiency, we compared the runtime between a direct set-indexed implementation and the proposed matrix-based formulation at increasing strain values. While both implementations simulate identical epidemiological dynamics, the set-based formulation exhibits rapidly increasing runtime with growing strain number. In contrast, the matrix-based implementation shows substantially slower runtime growth. For example, for the system of size $n=6$ with 256 compartments, the matrix-based formulation is approximately 76 times faster than the set-based formulation. For $n=10 \text{ and } n=11$ corresponding to 6,144 and 13,312 states respectively, the matrix framework achieves speedups exceeding 525 and 557 times faster, respectively with consistent runtime under one second. These results indicate that performance gains scale positively with model dimension.

\begin{figure}[H]
        \captionsetup[subfigure]{justification=centering}
        \centering
        \begin{subfigure}[b]{0.475\textwidth}
            \centering
            \includegraphics[width=\textwidth]{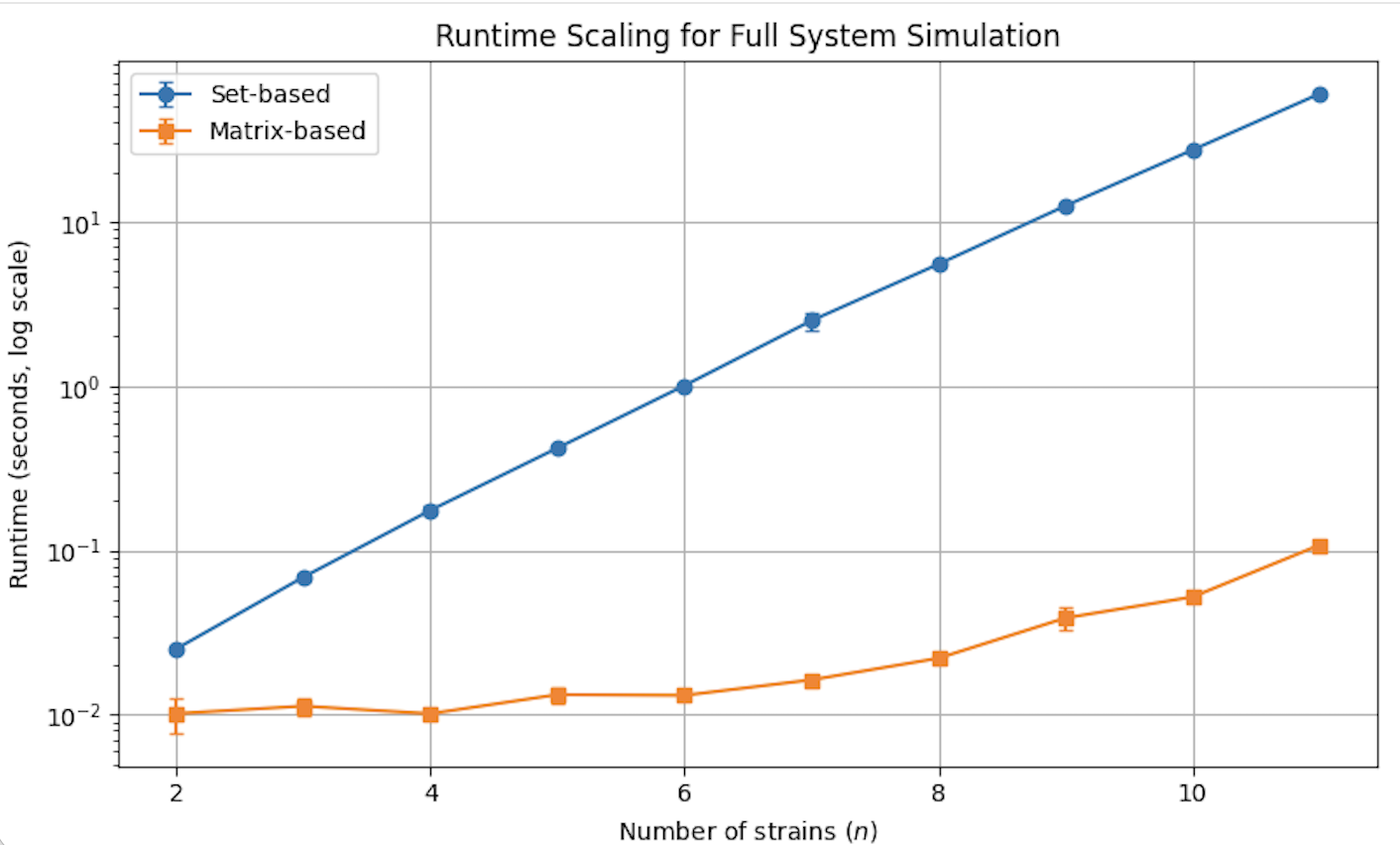}
            \caption{\small}
            \label{fig:runtime_a}
        \end{subfigure}
        \hfill
        \begin{subfigure}[b]{0.475\textwidth}
            \centering
            \includegraphics[width=\textwidth]{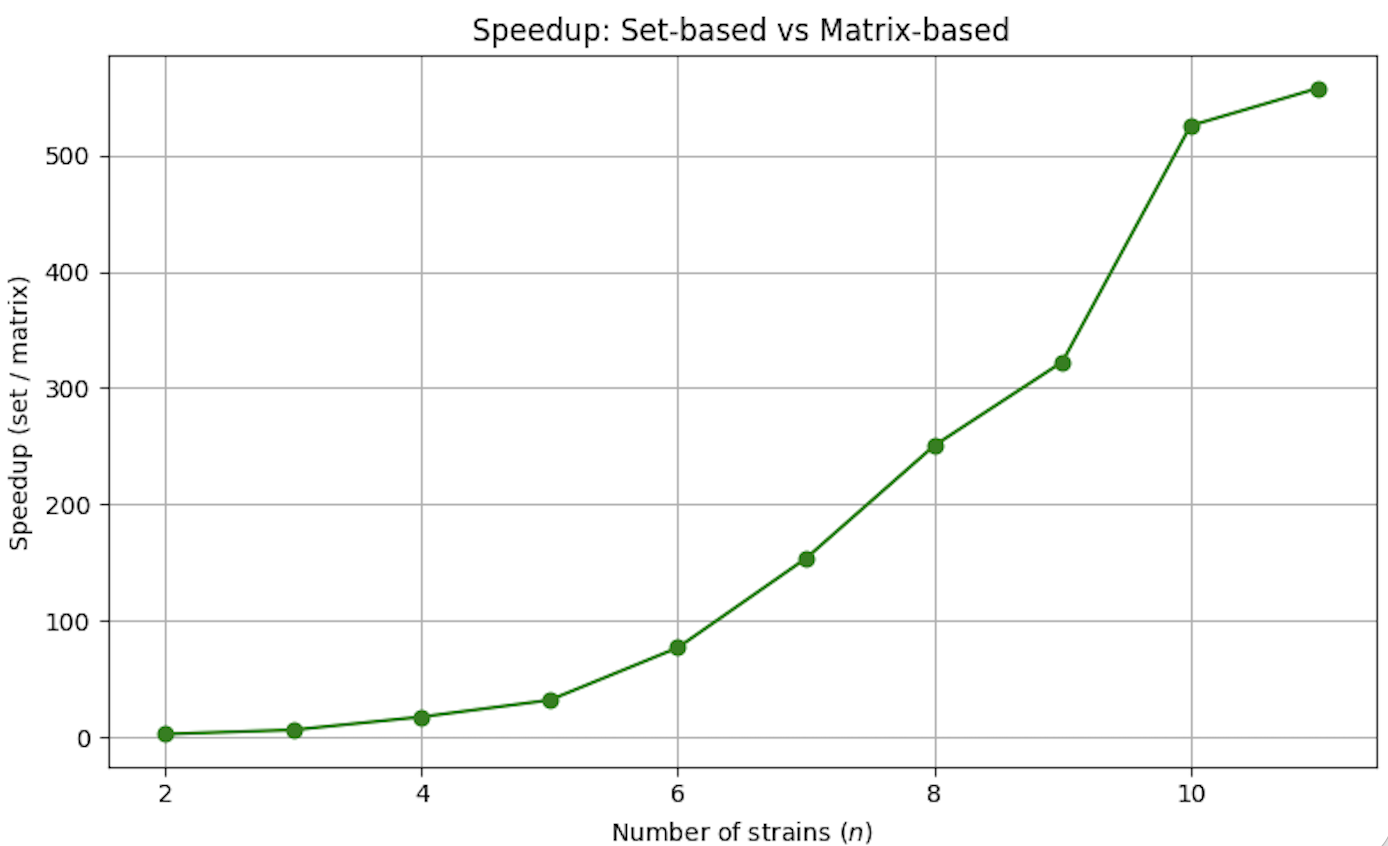}
            \caption{\small}
            \label{fig:runtime_b}
        \end{subfigure}
        
        \caption[]%
        {\small \textbf{Runtime simulation of set and matrix model.} (a) Comparison runtime (seconds, log scale) for full-system simulations as a function of the number of circulating strains $n$. Results compare a direct set-indexed implementation with the proposed matrix-based formulation. Points represent mean runtime across simulations, while small error bars indicate standard deviation. (b) Corresponding speedup factor defined as the ratio of set-based runtime to matrix-based runtime. The number of states in the system increases as $2^{n-1}(n+2)$, reflecting exponential growth in history-dependent states.} 
        \label{fig:runtime}
    \end{figure}

Figure~\ref{fig:runtime_a} illustrates runtime growth on a logarithmic scale. The set-based implementation exhibits near-exponential growth consistent with explicit iteration over strain histories, whereas the matrix-based formulation displays substantially flatter scaling behavior. Figure~\ref{fig:runtime_b} presents the corresponding speedup factor, demonstrating that relative efficiency increases systematically with strain number. 

Moreover, the matrix formulation does not reduce the dimensionality of the epidemiological system. Improvements arise from replacing nested iteration (set looping) in computation into vectorised linear algebra operations. Bringing this idea, the proposed framework substantially improves practical scalability while preserving epidemiological interpretation. These findings confirm that the matrix-based framework enables routine simulation of higher-dimensional multi-strain systems that would otherwise become computationally challenging under direct set-indexed implementations.

We have introduced a matrix-based numerical framework for modeling history-dependent multi-strain epidemic dynamics and demonstrated its applicability to strain replacement phenomena. By reformulating the traditional set-indexed Susceptible-Infected (SI) model into a matrix-based expression, this framework captures essential features of biological mechanisms such as partial cross-reactivity and sequential infection by specific strain while substantially simplifying numerical implementation. Note that the recovery compartment in this case is implicitly put into the system in the S compartment since the multi-strain infection allows each individual of having secondary, tertiary, and $n$-series of infection so they will remove from infected and go back into susceptible class.

This numerical framework provides several important insights into the dynamics of multi-strain pathogens addressing the limitations of set-based SI models. In term of computational work, this model is also computationally lightweight and improves runtime efficiency. In addition, the model also allows epidemic scenarios in different settings and is amenable to further development, since the system can be readily transformed to integer matrix where eventually easily to produce the output of numerical simulation.

One of the key advantages of this model is its ability to handle the simultaneous circulation of multiple strains. In real-world epidemics, it is common for multiple strains of a pathogen to co-circulate within a population, each interacting with the others through mechanisms such as cross-reactivity or cross-immunity \cite{kamo2002effect}. This interaction can lead to a variety of outcomes, including strain replacement or sequential waves of infection \cite{chisholm2020epidemiological}. By modeling these dynamics explicitly, the matrix representation framework provides a powerful tool for understanding how multi-strain pathogens behave under different conditions.

In summary, this numerical framework offers a novel, robust, and effective approach to model the dynamics of multi-strain pathogens. By capturing the effect of cross-immunity and infection histories, this approach enhances our ability to predict the outcomes of infectious disease outbreaks precisely. Future work should focus on refining the model by exploring more complex patterns of cross-immunity. Additionally, applying the model to real-world data from past epidemics can help validate its predictions and provide further insights into the behavior of multi-strain diseases. 

\section*{Data Availability}
All data and code used to produce the figures in this paper is publicly available at \url{https://github.com/fadilah-ilahi/multi-strain-framework}.

\section*{Acknowledgments}
FI is funded and sponsored by Indonesia Endowment Fund for Education Agency (LPDP Scholarship) from Ministry of Finance of Republic of Indonesia (Sponsor ID 202301220200205) during her doctoral program at The University of Manchester. TH and IH are supported by the Wellcome Trust (Grant number 227438/Z/23/Z). TH is supported by the Medical Research Council (Grant number UKRI483), and UKRI (Grant number MR/Z505328/1). The authors declare no competing interests.

\section*{Author Contributions}
\textbf{Conceptualisation:} FI, IH, TH.
\textbf{Methodology:} FI, IH, TH.
\textbf{Software:} FI.
\textbf{Supervision:} IH, TH.
\textbf{Writing -- original draft:} FI.
\textbf{Writing -- review and editing:} FI, IH, TH.

\printbibliography

@misc{WHOvariant,
	author = {{World Health Organisation}},
	title = {Tracking {SARS-CoV-2} variants},
	note = {https://www.who.int/activities/tracking-SARS-CoV-2-variants},
	year = {2026},
}

@article{dyson2021future,
	author = {Dyson, Louise and Hill, Edward M. and Moore, Sam and Curran-Sebastian, Jacob and Tildesley, Michael J. and Lythgoe, Katrina A. and House, Thomas and Pellis, Lorenzo and Keeling, Matt J.},
	journal = {Nature Communications},
	number = {1},
	pages = {5730},
	title = {Possible future waves of SARS-CoV-2 infection generated by variants of concern with a range of characteristics},
	volume = {12},
	year = {2021},
}

@article{kucharski2012original,
    author = {Kucharski, Adam J. AND Gog, Julia R.},
    journal = {PLOS Computational Biology},
    publisher = {Public Library of Science},
    title = {The Role of Social Contacts and Original Antigenic Sin in Shaping the Age Pattern of Immunity to Seasonal Influenza},
    year = {2012},
    month = {10},
    volume = {8},
    pages = {e1002741},
    number = {10},
}

@article{senaratne2020dengue,
  title={Dengue virus co-infections with multiple serotypes do not result in a different clinical outcome compared to mono-infections},
  author={Senaratne, UTN and Murugananthan, K and Sirisena, PDNN and Carr, JM and Noordeen, F},
  journal={Epidemiology \& Infection},
  volume={148},
  year={2020},
  publisher={Cambridge University Press}
}

@book{allen2008mathematical,
  title={Mathematical epidemiology},
  author={Allen, Linda JS and Brauer, Fred and Van den Driessche, Pauline and Wu, Jianhong},
  volume={1945},
  year={2008},
  publisher={Springer}
}

@article{kucharski2016capturing,
  title={Capturing the dynamics of pathogens with many strains},
  author={Kucharski, Adam J and Andreasen, Viggo and Gog, Julia R},
  journal={Journal of mathematical biology},
  volume={72},
  pages={1--24},
  year={2016},
  publisher={Springer}
}

@article{andreasen1997dynamics,
  title={The dynamics of cocirculating influenza strains conferring partial cross-immunity},
  author={Andreasen, Viggo and Lin, Juan and Levin, Simon A},
  journal={Journal of mathematical biology},
  volume={35},
  pages={825--842},
  year={1997},
  publisher={Springer}
}

@article{andreasen2023seasonality,
  title={Seasonality and the Coexistence of Pathogen Strains},
  author={Andreasen, Viggo and Dwyer, Greg},
  journal={The American Naturalist},
  volume={201},
  number={5},
  pages={639--658},
  year={2023},
  publisher={The University of Chicago Press Chicago, IL}
}

@article{wikramaratna2010effects,
  title={The effects of tertiary and quaternary infections on the epidemiology of dengue},
  author={Wikramaratna, Paul S and Simmons, Cameron P and Gupta, Sunetra and Recker, Mario},
  journal={PloS one},
  volume={5},
  number={8},
  pages={e12347},
  year={2010},
  publisher={Public Library of Science San Francisco, USA}
}

@article{champagne2019comparison,
  title={Comparison of stochastic and deterministic frameworks in dengue modelling},
  author={Champagne, Clara and Cazelles, Bernard},
  journal={Mathematical biosciences},
  volume={310},
  pages={1--12},
  year={2019},
  publisher={Elsevier}
}

@article{souto2019multiple,
  title={Multiple-serotype models of dengue virus transmission: simulation study and perspectives for the application of inference in epidemiological surveillance},
  author={Souto-Maior, Caetano},
  journal={bioRxiv},
  pages={583351},
  year={2019},
  publisher={Cold Spring Harbor Laboratory}
}

@article{kamo2002effect,
  title={The effect of cross-immunity and seasonal forcing in a multi-strain epidemic model},
  author={Kamo, Masashi and Sasaki, Akira},
  journal={Physica D: Nonlinear Phenomena},
  volume={165},
  number={3-4},
  pages={228--241},
  year={2002},
  publisher={Elsevier}
}

@article{chisholm2020epidemiological,
  title={Epidemiological consequences of enduring strain-specific immunity requiring repeated episodes of infection},
  author={Chisholm, Rebecca H and Sonenberg, Nikki and Lacey, Jake A and McDonald, Malcolm I and Pandey, Manisha and Davies, Mark R and Tong, Steven YC and McVernon, Jodie and Geard, Nicholas},
  journal={PLoS Computational Biology},
  volume={16},
  number={6},
  pages={e1007182},
  year={2020},
  publisher={Public Library of Science San Francisco, CA USA}
}

@article{lazebnik2022generic,
  title={Generic approach for mathematical model of multi-strain pandemics},
  author={Lazebnik, Teddy and Bunimovich-Mendrazitsky, Svetlana},
  journal={PloS one},
  volume={17},
  number={4},
  pages={e0260683},
  year={2022},
  publisher={Public Library of Science San Francisco, CA USA}
}

@article{Recker2009ImmunologicalDengue,
    title = {{Immunological serotype interactions and their effect on the epidemiological pattern of dengue}},
    year = {2009},
    journal = {Proceedings of the Royal Society B: Biological Sciences},
    author = {Recker, Mario and Blyuss, Konstantin B. and Simmons, Cameron P. and Hien, Tran Tinh and Wills, Bridget and Farrar, Jeremy and Gupta, Sunetra},
    number = {1667},
    month = {7},
    pages = {2541--2548},
    volume = {276},
    publisher = {Royal Society},
    doi = {10.1098/rspb.2009.0331},
    issn = {14712970},
    pmid = {19369266}
}

@article{Woodall2014PartialEpidemiology,
    title = {{Partial cross-enhancement in models for dengue epidemiology}},
    year = {2014},
    journal = {Journal of Theoretical Biology},
    author = {Woodall, Hannah and Adams, Ben},
    pages = {67--73},
    volume = {351},
    publisher = {Elsevier},
    url = {http://dx.doi.org/10.1016/j.jtbi.2014.02.016},
    doi = {10.1016/j.jtbi.2014.02.016},
    issn = {10958541},
    pmid = {24568780}
}

@article{martcheva2008vaccine,
  title={Vaccine-induced pathogen strain replacement: what are the mechanisms?},
  author={Martcheva, Maia and Bolker, Benjamin M and Holt, Robert D},
  journal={Journal of the Royal Society Interface},
  volume={5},
  number={18},
  pages={3--13},
  year={2008},
  publisher={The Royal Society London}
}

@article{mulberry2023pneumococcal,
  title={Pneumococcal population dynamics: Investigating vaccine-induced changes through multiscale modelling},
  author={Mulberry, Nicola and Rutherford, Alexander R and Colijn, Caroline},
  journal={PLoS Computational Biology},
  volume={19},
  number={12},
  pages={e1011755},
  year={2023},
  publisher={Public Library of Science San Francisco, CA USA}
}

@article{lythgoe2023lineage,
  title={Lineage replacement and evolution captured by 3 years of the United Kingdom Coronavirus (COVID-19) Infection Survey},
  author={Lythgoe, Katrina A and Golubchik, Tanya and Hall, Matthew and House, Thomas and Cahuantzi, Roberto and MacIntyre-Cockett, George and Fryer, Helen and Thomson, Laura and Nurtay, Anel and Ghafani, Mahan and others},
  journal={Proceedings of the Royal Society B},
  volume={290},
  number={2009},
  pages={20231284},
  year={2023},
  publisher={The Royal Society}
}

@article{ballesteros2009influenza,
  title={Influenza A gradual and epochal evolution: insights from simple models},
  author={Ballesteros, S{\'e}bastien and Vergu, Elisabeta and Cazelles, Bernard},
  journal={PLoS One},
  volume={4},
  number={10},
  pages={e7426},
  year={2009},
  publisher={Public Library of Science San Francisco, USA}
}

@article{boni2004influenza,
  title={Influenza drift and epidemic size: the race between generating and escaping immunity},
  author={Boni, Maciej F and Gog, Julia R and Andreasen, Viggo and Christiansen, Freddy B},
  journal={Theoretical population biology},
  volume={65},
  number={2},
  pages={179--191},
  year={2004},
  publisher={Elsevier}
}

@article{gog2002dynamics,
  title={Dynamics and selection of many-strain pathogens},
  author={Gog, Julia R and Grenfell, Bryan T},
  journal={Proceedings of the National Academy of Sciences},
  volume={99},
  number={26},
  pages={17209--17214},
  year={2002},
  publisher={National Academy of Sciences}
}

@article{wang2022multiscale,
  title={A multiscale model of COVID-19 dynamics},
  author={Wang, Xueying and Wang, Sunpeng and Wang, Jin and Rong, Libin},
  journal={Bulletin of Mathematical Biology},
  volume={84},
  number={9},
  pages={99},
  year={2022},
  publisher={Springer}
}

@article{li2022evaluating,
  title={Evaluating the impact of SARS-CoV-2 variants on the COVID-19 epidemic and social restoration in the United States: a mathematical modelling study},
  author={Li, Rui and Li, Yan and Zou, Zhuoru and Liu, Yiming and Li, Xinghui and Zhuang, Guihua and Shen, Mingwang and Zhang, Lei},
  journal={Frontiers in Public Health},
  volume={9},
  pages={801763},
  year={2022},
  publisher={Frontiers Media SA}
}

@article{majeed2022variant,
  title={Variant-specific interventions to slow down replacement and prevent outbreaks},
  author={Majeed, Bushra and Tosato, Marco and Wu, Jianhong},
  journal={Mathematical Biosciences},
  volume={343},
  pages={108703},
  year={2022},
  publisher={Elsevier}
}

@article{kim2022model,
  title={A model of COVID-19 pandemic with vaccines and mutant viruses},
  author={Kim, Young Rock and Choi, Yong-Jae and Min, Youngho},
  journal={Plos one},
  volume={17},
  number={10},
  pages={e0275851},
  year={2022},
  publisher={Public Library of Science San Francisco, CA USA}
}

@article{dutta2022covid,
  title={COVID-19 waves: variant dynamics and control},
  author={Dutta, Abhishek},
  journal={Scientific Reports},
  volume={12},
  number={1},
  pages={9332},
  year={2022},
  publisher={Nature Publishing Group UK London}
}

@article{Truscott2012EssentialInfluenza,
    title = {{Essential epidemiological mechanisms underpinning the transmission dynamics of seasonal influenza}},
    year = {2012},
    journal = {Journal of the Royal Society Interface},
    author = {Truscott, James and Fraser, Christophe and Cauchemez, Simon and Meeyai, Aronrag and Hinsley, Wes and Donnelly, Christl A. and Ghani, Azra and Ferguson, Neil},
    number = {67},
    month = {2},
    pages = {304--312},
    volume = {9},
    publisher = {Royal Society},
    doi = {10.1098/rsif.2011.0309},
    issn = {17425662},
    pmid = {21715400}
}

@techreport{Castillo-Chavez1989EpidemiologicalCross-immunity,
    title = {{Epidemiological models with age structure, proportionate mixing, and cross-immunity}},
    year = {1989},
    booktitle = {Journal of Mathematical Biology {\~{}})Spfinger-Verlag},
    author = {Castillo-Chavez, C and Hethcote, H W and Andreasen ]'2', V and Levin ]'2', S A and Liu ]'2', W M},
    pages = {233--258},
    volume = {27}
}

@article{Murall2014RevisingReplacement,
    title = {{Revising ecological assumptions about Human papillomavirus interactions and type replacement}},
    year = {2014},
    journal = {Journal of Theoretical Biology},
    author = {Murall, Carmen Lía and McCann, Kevin S. and Bauch, Chris T.},
    month = {6},
    pages = {98--109},
    volume = {350},
    publisher = {Academic Press},
    doi = {10.1016/j.jtbi.2013.12.028},
    issn = {10958541},
    pmid = {24412334}
}

@inproceedings{sayyar2022epidemic,
  title={Epidemic patterns of emerging variants with dynamical social distancing},
  author={Sayyar, Golsa and R{\"o}st, Gergely},
  booktitle={International Symposium on Mathematical and Computational Biology},
  pages={215--232},
  year={2022},
  organization={Springer}
}

@book{Martin2016BayesianFramework,
    title = {{Bayesian analysis with Python : unleash the power and flexibility of the Bayesian framework}},
    year = {2016},
    author = {Martin, Osvaldo.},
    publisher = {Packt Publishing},
    isbn = {9781785883804}
}

@article{haykal2024dynamics,
  title={Dynamics of SARS-CoV-2 Spike RBD Protein Mutation and Pathogenicity Consequences in Indonesian Circulating Variants in 2020--2022},
  author={Haykal, Nabiel Muhammad and Fadilah, Fadilah and Dewi, Beti Ernawati and Erlina, Linda and Prawiningrum, Aisyah Fitriannisa and Hegar, Badriul},
  journal={Genes},
  volume={15},
  number={11},
  pages={1468},
  year={2024},
  publisher={MDPI}
}

@article{prasetyoputri2022dynamics,
  title={The dynamics of circulating SARS-CoV-2 lineages in Bogor and surrounding areas reflect variant shifting during the first and second waves of COVID-19 in Indonesia},
  author={Prasetyoputri, Anggia and Dharmayanthi, Anik B and Iryanto, Syam B and Andriani, Ade and Nuryana, Isa and Wardiana, Andri and Ridwanuloh, Asep M and Swasthikawati, Sri and Hariyatun, Hariyatun and Nugroho, Herjuno A and others},
  journal={PeerJ},
  volume={10},
  pages={e13132},
  year={2022},
  publisher={PeerJ Inc.}
}

@article{mahroum2024vaccine,
  title={Vaccine-induced strain replacement: theory and real-life implications},
  author={Mahroum, Naim and Karaoglan, Birnur Sinem and Ulucam, Enes Sedat and Shoenfeld, Yehuda},
  journal={Future Microbiology},
  volume={19},
  number={11},
  pages={1017--1026},
  year={2024},
  publisher={Taylor \& Francis}
}

@misc{Python,
  title = {Python Documentation Programming Language},
  howpublished = {\url{https://www.python.org/}},
  note = {Accessed: 2026-07-23}
}

@article{wang2025optimizing,
  title={Optimizing a human monoclonal antibody for better neutralization of SARS-CoV-2},
  author={Wang, Qian and Guo, Yicheng and Casner, Ryan G and Yu, Jian and Nair, Manoj S and Ho, Jerren and Reddem, Eswar R and Mellis, Ian A and Wu, Madeline and Tzang, Chih-Chen and others},
  journal={Nature Communications},
  volume={16},
  number={1},
  pages={6195},
  year={2025},
  publisher={Nature Publishing Group UK London}
}

@article{wang2021antibody,
  title={Antibody resistance of SARS-CoV-2 variants B. 1.351 and B. 1.1. 7},
  author={Wang, Pengfei and Nair, Manoj S and Liu, Lihong and Iketani, Sho and Luo, Yang and Guo, Yicheng and Wang, Maple and Yu, Jian and Zhang, Baoshan and Kwong, Peter D and others},
  journal={Nature},
  volume={593},
  number={7857},
  pages={130--135},
  year={2021},
  publisher={Nature Publishing Group UK London}
}

@article{khoury2021neutralizing,
  title={Neutralizing antibody levels are highly predictive of immune protection from symptomatic SARS-CoV-2 infection},
  author={Khoury, David S and Cromer, Deborah and Reynaldi, Arnold and Schlub, Timothy E and Wheatley, Adam K and Juno, Jennifer A and Subbarao, Kanta and Kent, Stephen J and Triccas, James A and Davenport, Miles P},
  journal={Nature medicine},
  volume={27},
  number={7},
  pages={1205--1211},
  year={2021},
  publisher={Nature Publishing Group US New York}
}
\end{document}